\documentclass[12pt]{article}
\usepackage{amsmath}
\usepackage{amssymb}
\usepackage[latin1]{inputenc}
\usepackage{hyperref}

\newtheorem{corollary}{Corollary}

\newtheorem{lemma}{Lemma}
\newtheorem{theorem}{Theorem}

\newenvironment{proof}{\par\emph{Proof:~}}{\hfill$\square$\par}
\makeatletter
\def\@seccntformat#1{\csname the#1\endcsname.\ } % dots in Section headings
\makeatother

\title{On the binary codes with parameters of triply-shortened $1$-perfect codes%
\thanks{The work was supported by the Federal Target Grant
``Scientific and Educational Personnel of Innovation Russia''
for 2009-2013 (government contract No. 02.740.11.0429) and the
Russian Foundation for Basic Research (grant 10-01-00424).}
}
\author{\href{http://arXiv.org/a/krotov_d_1}{Denis S. Krotov}%
\thanks{Sobolev Institute of Mathematics, Novosibirsk, Russia. \texttt{krotov@math.nsc.ru}}
\thanks{Mechanics and Mathematics Department, Novosibirsk State University, Russia}
}
\date{}

\begin{document}
\maketitle
\begin{abstract}
We study properties of binary codes with parameters close
to the parameters of $1$-perfect codes.
An arbitrary binary $(n=2^m-3, 2^{n-m-1}, 4)$ code $C$,
i.e., a code with parameters
of a triply-short\-ened extended Hamming code,
is a cell of an equitable partition of the $n$-cube into six cells.
An arbitrary binary $(n=2^m-4, 2^{n-m}, 3)$ code $D$, i.e., a code
with parameters of a triply-short\-ened Hamming code,
is a cell of an equitable family (but not a partition) from six cells.
As a corollary, the codes $C$ and $D$ are completely semiregular;
i.e., the weight distribution of such a code depends only
on the minimal and maximal codeword weights and the code parameters.
Moreover,
if $D$ is self-complementary, then it is completely regular.

As an intermediate result, we prove,
in terms of distance distributions,
 a general criterion for
a partition of the vertices of a graph
(from rather general class of graphs, including the distance-regular graphs) to be equitable.
\providecommand\keywords{\par \textbf{Keywords:}~}
\providecommand\subclass{\par \textbf{MSC:}~}
\keywords{$1$-perfect code; triply-shortened $1$-perfect code; equitable partition; perfect coloring; weight distribution; distance distribution}
\subclass{94B25}
\end{abstract}
% We study properties of binary codes with parameters close to the parameters of $1$-perfect codes. An arbitrary binary $(n=2^m-3, 2^{n-m-1}, 4)$ code $C$, i.e., a code with parameters of a triply-shortened extended Hamming code, is a cell of an equitable partition of the $n$-cube into six cells. An arbitrary binary $(n=2^m-4, 2^{n-m}, 3)$ code $D$, i.e., a code with parameters of a triply-shortened Hamming code, is a cell of an equitable family (but not a partition) from six cells. As a corollary, the codes $C$ and $D$ are completely semiregular; i.e., the weight distribution of such a code depends only on the minimal and maximal codeword weights and the code parameters. Moreover, if $D$ is self-complementary, then it is completely regular.
% As an intermediate result, we prove, in terms of distance distributions, a general criterion for a partition of the vertices of a graph (from rather general class of graphs, including the distance-regular graphs) to be equitable.
% Keywords: $1$-perfect code; triply-shortened $1$-perfect code; equitable partition; perfect coloring; weight distribution; distance distribution
\section{{Introduction}}

In this paper, we prove some regular properties
of the binary codes with parameters of
triply-shortened (extended) Hamming code.
We will see that these codes
have more commonality with the class of perfect codes
than simply optimality and close parameters.
The subject and approach
have a similarity with the previous paper about the
doubly-shortened case \cite{Kro:2m-3}, but there are some new essentials.
At first, for describing all results,
we need to generalize
the concept of equitable partition,
leaving it rather strong
to inherit the main algebraic-combinatorial properties.
At second, we derive, as corollaries,
 new properties of the considered class of codes,
such as some weaker variant of complete regularity.
At third, we prove a general criterion
on equitability of a partition,
whose usability is not bounded by the current research.
Some properties of the codes with considered parameters were found in \cite{KOP:2011}
and utilized there for classification of codes with small parameters.

\def\3#1{{\widetilde{#1}}} % my_macros
\def\0{\3{0}}\def\1{\3{1}}\def\2{\3{2}} % my_macros

We call a collection $\mathbf{P}=(P_0, P_1, \ldots, P_{r-1})$ of vertex subsets ({\em cells}) of a simple graph $G=(V,E)$
(in this paper, a binary Hamming graph, or a hypercube)
 an {\em equitable family} if there is a matrix $(s_{ij})_{i,j=0}^{r-1}$ (the {\em quotient} matrix) such
that any vertex $\bar x$ has exactly $\sum_{i\in \mathbf{i}(\bar x)} s_{ij}$ neighbors
from $P_j$ for every $j=0,1,...,r-1$
where $\mathbf{i}(\bar x)=\{i\mid \bar x \in P_i\}$.
If $P_0, P_1, \ldots, P_{r-1}$ are mutually disjoint
and cover  whole $V$, then $\mathbf{P}$ is known
as an {\em equitable partition}.

Famous examples of equitable partitions in regular graphs are
$1$-perfect codes (together with their complements).
In the case of a hypercube, the corresponding quotient matrix is $((0,n)(1,n{-}1))$
and the parameters of a code are $(n=2^m-1, 2^n-m, 3)$ (the code length,
or the hypercube dimension; the cardinality;
the minimal distance between codewords). Trivially, such codes are optimal, i.e.,
have the maximum cardinality for given length and code distance.
As shown in \cite{BesBro77}, any $(n=2^m-1-t, 2^{n-m}, 3)$ code is also optimal for
$t=1,2,3$.
For short, the parameters
$(n=2^m-1-t, 2^{n-m}, 3)$ and $(n=2^m-t, 2^{n-m-1}, 4)$, $t=0,1,2,3$
will be referred to as
\def\opt{_{\scriptscriptstyle \mathrm{op}}}
$(n, 3)\opt$, $(n, 3)'\opt$, $(n, 3)''\opt$, $(n, 3)'''\opt$ and
$(n, 4)\opt$, $(n, 4)'\opt$, $(n, 4)''\opt$, $(n, 4)'''\opt$, respectively.

Every  $(n, 3)'\opt$ % $(2^m-2, 3)\opt$ % $(n=2^m-2, 2^{n-m}, 3)$
code is indeed a {\em shortened} $1$-perfect
$(n+1, 3)\opt$ % $(n=2^m-1, 2^n-m, 3)$
code \cite{Bla99}, i.e.,
can be obtained from a $1$-perfect code by fixing one coordinate.
Moreover, it can be seen that every
$(n, 3)'\opt$ % $(n=2^m-2, 3)\opt$ % $(n=2^m-2, 2^{n-m}, 3)$
code is a cell of an
equitable partition with quotient matrix $((0,n,0)(1,n{-}2,1)(0,n,0))$.

The situation with $(n, 3)''\opt$ % $(2^m-3, 3)\opt$ % $(n=2^m-3, 2^{n-m}, 3)$
is different.
There are
such codes that cannot be represented as doubly-shortened $1$-perfect
% at least two $(13,512,3)$ codes that are not
% doubly-shortened $(15,2048,3)$
\cite{OstPot:13-512-3,KOP:2011}.
Nevertheless, every $(n, 3)''\opt$ % $(n=2^m-3, 3)\opt$ % $(n=2^m-3, 2^{n-m}, 3)$
code is a cell of
an equitable partition with quotient matrix\linebreak[4]
$((0,1,n{-}1,0)(1,0,n{-}1,0)(1,1,n{-}4,2)(0,0,n{-}1,1))$ \cite{Kro:2m-3}.

Our current topic is the case of $(n, 3)'''\opt$. % $(2^m-4, 3)\opt$. % $(n=2^m-4, 2^{n-m}, 3)$.
For these parameters, examples of codes that are not
triply-shortened $1$-perfect
are also known \cite{OstPot:13-512-3,KOP:2011}.
Moreover, for $n\geq 12$
there are $(n, 3)'''\opt$ % $(2^m-4, 3)\opt$ % $(n=2^m-4, 2^{n-m}, 3)$
codes that cannot be represented
as a cell of an equitable partition, because such codes are not
distance invariant in general (by shortening a nonlinear $1$-perfect
code, it is possible to obtain an
$(n, 3)'''\opt$ % $(2^m-4, 3)\opt$ % $(n=2^m-4, 2^{n-m}, 3)$
code whose
weight distribution with respect to a code vertex depends on the
choice of this vertex).
We state that, nevertheless, such a code is
a cell of some generalization of an equitable partition (equitable
family), which inherit the main algebraic properties
of equitable partitions.
Moreover, if we extend such a code to an
$(n+1, 4)'''\opt$ % $(2^m-3, 4)\opt$
%(i.e., $(n=2^m-3, 2^{n-m+1}, 4)$)
code, by adding the parity-check bit,
then the code obtained will be a cell of an equitable partition.
As a corollary, we derive some variant of distance invariance for the codes with considered parameters.

We start with distance-$4$ codes.
In Section~\ref{s:dd}, we consider an arbitrary
$(n, 4)'''\opt$ % $(2^m-3, 4)\opt$ % $(n=2^m-3,2^{n-m-1}, 4)$
code $C_0$, define the other
five cells of the generated partitions, and prove
that the mutual distance distribution of the partition cells
does not depend on the choice of the code.
In Section~\ref{s:crit},
we prove rather general criterion for a partition
of the vertices of a graph to be equitable.
In Section~\ref{s:main}, we use this criterion to show that
the partition generated by $C_0$ is equitable; as a corollary,
we derive that any $(n, 3)'''\opt$ % $(2^m-4, 3)\opt$ % $(n=2^m-4, 2^{n-m}, 3)$
code also generates an
equitable family.
In Section~\ref{s:reg}, we prove some weak form of complete regularity
for the distance-$3$ and distance-$4$ codes with considered parameters and
the distance invariance for the distance-$4$ codes.
In the last section, we mention two other interesting properties of the considered classes of codes,
one of which was proved earlier in the paper \cite{KOP:2011}.

%=====================================================================
\section{Generated subsets and distance distributions}\label{s:dd}

The $n$-dimensional hypercube graph will be denoted by
$H^n = (V(H^n),E(H^n))$.
Recall, that $V(H^n)$ consists of the words of length $n$
in the alphabet $\{0,1\}$,
two words being adjacent if and only if
they differ in exactly one position.
By $d(\cdot,\cdot)$ we denote the natural graph distance
in $H^n$ (Hamming distance); by $\overline 0$ and $\overline 1$, the
all-zero and all-one words respectively.
The graph $H^n$ is bipartite, and we denote its parts by $V_{\mathrm{ev}}$ and $V_{\mathrm{od}}$,
$V_{\mathrm{ev}}$ containing $\overline 0$.

Let $C_0$ be an $(n, 4)'''\opt$ % $(n=2^m-3, 4)\opt$ % $(n=2^m-3,2^{n-m-1}, 4)$
code.
As proved in \cite{KOP:2011} (see Lemma~\ref{l:0} below),
the mutual distances between the codewords of $C_0$ are even;
i.e., either
$C_0\subset V_{\mathrm{ev}}$ or
$C_0\subset V_{\mathrm{od}}$.
We assume the former.
Define
\begin{eqnarray}
 C_\0 &=& C_0 + \overline 1 , \label{eq:C0def} \\
 C_1 &=& \{\bar x \mid d(\bar x, C_\0) = 1,\ \bar x\not\in C_0 \} , \\
 C_\1 &=& \{\bar x \mid d(\bar x, C_0) = 1,\ \bar x\not\in C_\0 \} = C_1 + \overline 1 , \\
 C_2 &=& V_{\mathrm{ev}}\setminus (C_0 \cup C_1),\\
 C_\2 &=& V_{\mathrm{od}}\setminus (C_\0 \cup C_\1)= C_2 + \overline 1. \label{eq:C2def}
\end{eqnarray}
For convenience, we will associate $\0$, $\1$ and $\2$
with the numbers $3$, $4$ and $5$.
So, $(C_i)_{i=0}^{\2}$ is a partition of $V(H^n)$,
while $(C_0,C_1,C_2)$ and $(C_\0,C_\1,C_\2)$
are partitions of $V_{\mathrm{ev}}$ and $V_{\mathrm{od}}$ respectively.
Denote
$$ A^j_l (\bar x) = |\{ \bar y \in C_j \mid d(\bar x, \bar y) = l \}|,
\qquad j\in \{0,1,2,\0,\1,\2\}, \ \bar x \in V(H^n); $$
the
{$(n{+}1)$}-tuple $(A^j_0 (\bar x),A^j_1 (\bar x),\ldots ,A^j_n
(\bar x))$ is known as the {\em weight distribution} of $C_j$ with respect
to $\bar x$;

\def\A{\overline{A}\vphantom{A}}
\def\B{\overline{B}\vphantom{B}}
$$ \A^{ij}_l = \frac 1{| C_i|} \sum_{\bar x \in C_i} A^j_l(\bar x),
\qquad i,j\in \{0,1,2,\0,\1,\2\};$$ the collection
${(}(\A^{ij}_l{)_{i,j=0}^\2)_{l=0}^n}$ will be referred to as the
{\em distance distribution} of $(C_i)_{i=0}^\2$; the {$(n{+}1)$}-tuple
$(\A^{ii}_0,\A^{ii}_1,\ldots ,\A^{ii}_n)$ is known as the {\em inner
distance distribution} of $C_i$.

As noted in \cite{BesBro77}, there are more than one possibility for the
inner distance distribution of an $(n,3)'''\opt$ code.
However, the ``extended'' variant of
the proof of \cite[Theorem 6.1]{BesBro77}
provides us with the following key statement:

\begin{lemma}[\cite{KOP:2011}]\label{l:0}
The inner distance distribution of an
$(n, 4)'''\opt$  % $(2^m-3,4)\opt$ % $(n=2^m-3,2^{n-m-1},4)$
code $C_0$ does not depend on the choice of the code.
In particular, $\A^{00}_{n-1}=1$ and $\A^{00}_{i}=0$ for odd $i$.
\end{lemma}

It is not difficult to expand this fact
to all the coefficients ${(}(\A^{ij}_l{)_{i,j=0}^\2)_{l=0}^n}$:

\begin{lemma}\label{l:012345}
The distance distribution ${(}(\A^{ij}_l{)_{i,j=0}^\2)_{l=0}^n}$ of
$(C_i)_{i=0}^\2$ does not depend on the choice of the
$(n, 4)'''\opt$ % $(2^m-3,4)\opt$ % $(n=2^m-3,2^{n-m-1},4)$
code $C_0$.
\end{lemma}
\begin{proof}
Since, because of the code distance, every vertex of $C_0$ has
not more than one neighbor from $C_\0$,
we find from $\A^{0\0}_{1}=\A^{00}_{n-1}=1$ that it has exactly one such neighbor. And vise versa, every vertex of $C_\0$ has exactly one neighbor from $C_0$.
Then, from the definitions of $C_i$ and $\A^{ij}_l$,
we have, for every $i\in \{0,1,2,\0,\1,\2\}$,
\begin{eqnarray*}
\A^{i1}_l &=& (n-l+1)\cdot \A^{i\0}_{l-1} +(l+1)\cdot \A^{i\0}_{l+1} - \A^{i0}_{l}, \\
\A^{i2}_l &=& \left({n \atop l}\right) - \A^{i0}_{l} - \A^{i1}_{l}, \qquad \mbox{$l$ even if $i\in \{0,1,2\}$, $l$ odd if $i\in \{\0,\1,\2\}$, } \\
\A^{ij}_l &=& \A^{i\,\3{j}}_{n-l} \qquad \forall j\in\{0,1,2\}, \\
|C_i|\cdot \A^{ij}_l & = & |C_j|\cdot \A^{ji}_l \qquad \forall j\in\{0,1,2,\0,\1,\2\}
\end{eqnarray*}
(see the similar \cite[Lemma~3]{Kro:2m-3} for details).
Using these formulas and starting from $(\A^{00}_l)_l$, we can derive
$(\A^{ij}_l)_l$ for every $i,j\in \{0,1,2,\0,\1,\2\}$.
\end{proof}
As we will see in Section~\ref{s:reg}, even the weight distribution
$(A^j_l(\bar x))_{l=0}^n$ depends only on $j$ and $i$ such that $\bar x\in C_i\cap C_j$, and does not depend on the choice of $C_0$ or $\bar x$ from $C_i\cap C_j$.
But now we have only the distance distribution
and we have to derive
from this knowledge that the partition is equitable.
It turns out, there is a general fact
connecting the distance distribution
of a partition with its equitability,
and this is the topic of the next section.

%=================================================================================

\section{A criterion on equitability}\label{s:crit}
We will formulate a criterion on equitability of partitions
in quite general class of graphs, including so-called distance-regular graphs. For the hypercube, the parameters $\gamma$ and $\delta$ in the following lemma equal $0$ and $2$ respectively.
\begin{lemma}\label{l:equi}
Let $G=(V(G),E(G))$ be a simple graph.
Assume that there are two constants $\gamma$ and $\delta$
such that, in $G$, every two adjacent vertices have $\gamma$
common neighbors and every two non-adjacent vertices have
$0$ or $\delta$ common neighbors.
Let $\mathbf{C}=(C_0,\ldots,C_k)$ be a partition of
$V(G)$ with distance distribution
$((\A^{ij}_l)_{i,j=0}^k)_{l=0}^n$. Then the following three statements are equivalent:
\begin{itemize}
\item[\rm (a)] The partition $\mathbf{C}$ is equitable.

\item[\rm (b)] The numbers $\A^{ji}_1$ and $\A^{ii}_2$ satisfy
$$
|C_i|( \gamma\A^{ii}_1 +\delta\A^{ii}_2) = \sum_{j=0}^k |C_j|\cdot \A^{ji}_1(\A^{ji}_1-1)\qquad \forall i\in\{0,...,k\}.
$$
\item[\rm (c)] There is at least one equitable partition of $V(G)$
with the same numbers
  $\A^{ij}_1$ and $\A^{ii}_2$, $i,j=0,...,k,$ in the distance distribution.
\end{itemize}
\end{lemma}
\begin{proof}
(a)$\Leftrightarrow$(b)
Let us calculate in two ways the number $R$ of triples $(\bar x,\bar y, \bar z)$
of vertices such that
$\bar x,\bar z \in C_i$ are different neighbors of $\bar y$. If we choose
$\bar x$, then $\bar z$, and then $\bar y$, then we have
\begin{equation}\label{eq:R-xzy}
R = \sum_{\bar x\in C_i}
(A^i_1(\bar x) \cdot \gamma
+A^i_2(\bar x) \cdot \delta) =|C_i|
( \gamma\A^{ii}_1 +\delta\A^{ii}_2)
\end{equation}
choices.
If we choose $\bar y$ and then $\bar x$ and $\bar z$,
then the number of choices is
\begin{equation}\label{eq:R-yxz}
R = \sum_{\bar x \in V(G)} A^i_1(\bar x)(A^i_1(\bar x)-1)
= \sum_{j=0}^k \sum_{\bar x \in C_j} A^i_1(\bar x)(A^i_1(\bar x)-1)
\end{equation}
Comparing (\ref{eq:R-xzy}) and   (\ref{eq:R-yxz})
and using the Cauchy--Bunyakovsky inequality, we get
$$
|C_i|
( \gamma\A^{ii}_1 +\delta\A^{ii}_2) = \sum_{j=0}^k \sum_{\bar x \in C_j} A^i_1(\bar x)(A^i_1(\bar x)-1)
\geq
\sum_{j=0}^k |C_j|\cdot \A^{ji}_1(\A^{ji}_1-1)
$$
which holds with equality for all $i$ if and only if for all $i,j\in\{0,...,k\}$
and $\bar x\in C_j$ the value $A^i_1(\bar x)$ equals to its average value over $C_j$.
Since the last obviously coincides with the definition of an equitable partition,
(a) and (b) are equivalent.

(c)$\Rightarrow$(a) readily follows from (a)$\Leftrightarrow$(b);
(a)$\Rightarrow$(c) is trivial.
\end{proof}

%=================================================================================
\section{Main results}\label{s:main}

We are now ready to prove the main results of
our research, namely, the equitability of the partition
generated by an
$(n, 4)'''\opt$ % $(2^m-3, 4)\opt$ % $(n=2^m-3,2^{n-m-1}, 4)$
code and of the family of subsets generated by an
$(n, 3)'''\opt$ % $(2^m-4, 3)\opt$ % $(n=2^m-4,2^{n-m}, 3)$
code.

\begin{theorem}\label{th:main4}
Let $C_0$ be an
$(n, 4)'''\opt$ % $(n=2^m-3, 4)\opt$ % $(n=2^m-3,2^{n-m-1}, 4)$
code.
Then $C_0$ together with the related sets
$C_1,C_2,C_\0,C_\1,C_\2$ defined by
{\rm (\ref{eq:C0def})--(\ref{eq:C2def}) }
form an equitable partition with
quotient matrix
\begin{equation}\label{eq:quo_ext}
S=\left(
\begin{array}{@{\ \ }c@{\ \ }c@{\ \ }r@{\ \ \ \ \ }c@{\ \ }c@{\ \ }r@{\ \ }}
 0&0&0&1&n{-}1&0 \\
 0&0&0&1&n{-}4&3 \\
 0&0&0&0&n{-}1&1 \\[0.4ex]
 1&n{-}1&0&0&0&0 \\
 1&n{-}4&3&0&0&0 \\
 0&n{-}1&1&0&0&0
\end{array}\right)
\end{equation}
\end{theorem}
\begin{proof}
By Lemmas~\ref{l:012345} and~\ref{l:equi}, it is sufficient to prove
the statement for some
$(n, 4)'''\opt$ % $(2^m-3, 4)\opt$ % $(n=2^m-3,2^{n-m-1}, 4)$
code, say, the triply-shortened
extended Hamming code. Indeed, it is easy to check for any triply-shortened
extended $1$-perfect code. For such a code $C_0$, there are seven codes
$C_{001}$, $C_{010}$, $C_{100}$, $C_{110}$, $C_{101}$, $C_{011}$, $C_{111}$
such that the code
$$ C=C_0 000 \cup C_{001}001 \cup C_{010}010 \cup C_{100}100
\cup C_{110}110 \cup C_{101}101 \cup C_{011}011 \cup C_{111}111
$$
is extended $1$-perfect.
Then from the well-known property $C=C+\overline 1$ and from definitions
we derive $C_\0=C_{111}$, $C_\2=C_{001} \cup C_{010} \cup C_{100}$,
 $C_2=C_{110} \cup C_{101} \cup C_{011}$. Now, it is straightforward to check
from the definition of a $1$-perfect code that the partition
$(C_0,C_1,C_2,C_\0,C_\1,C_\2)$ is equitable with quotient matrix (\ref{eq:quo_ext}), see
the similar \cite[Proposition~1]{Kro:2m-3}.
\end{proof}
%=================================================================================
% \section{Main result: distance-$3$ codes}

\begin{theorem}\label{th:main3}
Let $D_0$ be an
$(n, 3)'''\opt$ % $(n=2^m-4,3)\opt$ % $(n=2^m-4,2^{n-m},3)$
code and let the sets
$D_1$, $D_2$, $D_\0$, $D_\1$, $D_\2$ be defined as
\begin{eqnarray}
D_1 & = & \{ \bar x\in V(H^n) \mid d(\bar x,C_0)=1 \} \label{eq:D0def} \\
D_2 & = & \{ \bar x\in V(H^n) \mid d(\bar x,C_0)>1 \} \label{eq:D1def} \\
D_{\3i} &=& D_i + \overline 1, \qquad i=0,1,2.  \label{eq:D2def}
\end{eqnarray}
Then the collection
 $(D_0,D_1,D_2,D_\0,D_\1,D_\2)$ is an equitable family with the
quotient matrix
$$
S=\left(
\begin{array}{@{\ \ }c@{\ \ }c@{\ \ }r@{\ \ \ \ \ }c@{\ \ }c@{\ \ }r@{\ \ }}
 0&n&0&0&0&0 \\
 1&n{-}4&3&0&0&0 \\
 0&n{-}2&2&0&2&-2 \\[0.4ex]
 0&0&0&0&n&0 \\
 0&0&0&1&n{-}4&3 \\
 0&2&-2&0&n{-}2&2
\end{array}\right)
$$
\end{theorem}
\begin{proof}
We have to prove that, for every
 $i,j\in \{0,1,2\}$ and $k\in \{0,1,2,\0,\1,\2\}$, the number of
vertices of $D_k$ adjacent to a fixed vertex $\bar x\in D_i \cap D_\3{j}$ does not
depend on the choice of $\bar x$ (as well as on the choice of the
initial code $D_0$) and is defined by the following table:
\begin{equation}\label{eq:tab}
(T_{i\,\3j,k}):\quad
\begin{array}{c|@{\ \ }c@{\ \ }c@{\ \ }c@{\ \ \ }c@{\ \ }c@{\ \ }c@{\ \ }|}
& 0 & 1 & 2 & \0 & \1 & \2 \\
\hline
0{\0} & 0&n&0&0&n&0 \\
0{\1} & 0&n&0&1&n{-}4&3 \\
1{\0} & 1&n{-}4&3&0&n&0 \\
1{\1} & 1&n{-}4&3&1&n{-}4&3 \\
1{\2} & 1&n{-}2&1&0&n{-}2&2 \\
2{\1} & 0&n{-}2&2&1&n{-}2&1 \\
2{\2} & 0&n&0&0&n&0 \\ \hline
\end{array}
\end{equation}
  Indeed, for $i\,\3{j}\in\{0\0,0\1,1\0,1\1,1\2,2\1,2\2\}$, the sum of the $i$th and $\3{j}$th rows of the matrix $S$
  coincides with the corresponding row of the table (\ref{eq:tab}).
There are no rows indexed by $0{\2}$ or
$2{\0}$ in the table (\ref{eq:tab}) because, as we will see below
(table (\ref{eq:tbl})),
the intersection of $D_0$ and $D_\2$, as well as $D_\0$ and $D_2$, is empty.

Now, let $C_0$ be the
$(n,4)'''\opt$ % $(2^m-3,4)\opt$ % $(n'=2^m-3, 2^{n'-m-1},4)$
code obtained from
$D_0$ by appending the parity-check bit to every codeword.
Let the partition $\mathbf{C}=(C_0,\linebreak[1]C_1, \linebreak[1]C_2, \linebreak[1]C_\0, \linebreak[1]C_\1, \linebreak[1]C_\2)$
be defined by
(\ref{eq:C0def})--(\ref{eq:C2def}).
It is straightforward from the definitions of $C_i$ and $D_i$ that
for any vertex $\bar x$ the indexes $i$ and $\3j$ such that
$\bar x \in D_i \cap D_{\3 j}$ can be derived from the knowledge of
cells from $\mathbf{C}$ that contain $\bar x 0$ and $\bar x 1$:
\begin{equation}\label{eq:tbl}
\begin{tabular}{c|c|c}
$\bar x 0 \in$ or $\bar x 1 \in$ & $\bar x 1 \in$ or $\bar x 0 \in$
& $\bar x \in$ \\
\hline
$C_0$ & $C_\0$ & $D_0 \cap D_\0$ \\
$C_0$ & $C_\1$ & $D_0 \cap D_\1$ \\
$C_1$ & $C_\0$ & $D_1 \cap D_\0$ \\
$C_1$ & $C_\1$ & $D_1 \cap D_\1$ \\
$C_1$ & $C_\2$ & $D_2 \cap D_\1$ \\
$C_2$ & $C_\1$ & $D_1 \cap D_\2$ \\
$C_2$ & $C_\2$ & $D_2 \cap D_\2$ \\ \hline
\end{tabular}
\end{equation}
Note that the case $\bar x0 \in C_0$, $\bar x1 \in C_\2$
or similar is impossible, because by Theorem~\ref{th:main4}
an element of $C_0$ has no neighbors in $C_\2$ (i.e., the $0\2$th element of the matrix $S$ in (\ref{eq:quo_ext}) equals $0$).

{\it Observation} (*): \\
$\bar x\in D_0$ if and only if $\bar x 0 \in C_0$ or $\bar x 1 \in C_0$; \\
$\bar x\in D_\0$ if and only if $\bar x 0 \in C_\0$ or $\bar x 1 \in C_\0$; \\
$\bar x\in D_2$ if and only if $\bar x 0 \in C_\2$ or $\bar x 1 \in C_\2$; \\
$\bar x\in D_\2$ if and only if $\bar x 0 \in C_2$ or $\bar x 1 \in C_2$.\\
(From this observation, one can note that there is no strict synchronization between the enumerations of $C_{...}$ and $D_{...}$.)

Now assume, for example,
that $\bar x 0 \in C_1$ and $\bar x 1 \in C_\2$.
By Theorem~\ref{th:main4},
$\bar x 0$ has exactly $3$ neighbors in $C_\2$.
One of them is $\bar x 1$ and the other two have the form $\bar y 0$.
Taking into account observation (*) and the fact that $\bar x 0$
has no neighbors from $C_\2$ because of its unparity, we conclude
that $\bar x$ has exactly $2$ neighbors from $D_2$.
Since $\bar x 0$ has exactly one neighbor in $C_\0$, we also see
that $\bar x$ has exactly one neighbor from $D_\0$.
Similarly, considering the neighborhood of $\bar x 1$ and using
Theorem~\ref{th:main4} and observation~(*), we find that
$\bar x$ has no neighbors in $D_0$ and exactly one neighbor in
$D_\2$. The numbers of neighbors in $D_1$ and in $D_\1$ are calculated
automatically as $n-0-2$ and $n-1-1$ respectively. So, the $1\2$th
line of the table $(T_{i\,\3j,k})$ is confirmed for the vertex $\bar x$.

The other cases can be easily checked by the same way, and there is
no need to duplicate the same arguments with the only difference in
table values.
\end{proof}
%===========================================================================%=======================================================
\section{Regularity and weight distributions} \label{s:reg}
A code is called
\emph{distance invariant} if its weight distribution with respect to
any codeword does not depend on the choice of the codeword.
 A code is called
\emph{completely regular} if its weight distribution with respect to
some initial vertex depends only on the distance between the initial
vertex and the code. We call a code \emph{completely semiregular} if
its weight distribution with respect to some initial vertex $\bar x$ depends
only on the distance between $\bar x$  and the code and the
distance between $\bar x+\overline 1$  and the code.
\begin{corollary}
{\rm (a)} Any $(n,3)'''\opt$ % $(2^m-4,3)\opt$ % $(n=2^m-4, 2^{n-m}, 3)$
code is completely semiregular.
Any $(n,4)'''\opt$, % $(2^m-3,4)\opt$,  % $(n=2^m-3, 2^{n-m-1}, 4)$,
$(n,3)''\opt$, % $(2^m-3,3)\opt$, % $(n=2^m-3, 2^{n-m}, 3)$,
or $(n,4)''\opt$ % $(2^m-2,4)\opt$ % $(n=2^m-2, 2^{n-m-1}, 4)$
code is completely semiregular and distance invariant.

{\rm (b)} Any self-complementary (i.e., $C_0=C_0+\overline 1$) code with parameters
$(n,3)'''\opt$ % $(2^m-4,3)\opt$  % $(n=2^m-4, 2^{n-m}, 3)$
is completely regular.
\end{corollary}
The last statement can be treated as that any
$(n=2^m-4, 2^{n-m-1}, 3)$
code
in the \emph{folded} hypercube graph of degree $n$ (the graph obtained by merging the antipodal pairs of vertices) is completely regular.
\begin{proof}
Let $D_0$ be an
$(n,3)'''\opt$ % $(2^m-4,3)\opt$ % $(n=2^m-4, 2^{n-m}, 3)$
code,
and let $\chi$ be the characteristic vector-function of
its generated equitable family $(D_0,D_1,D_2,D_\0,D_\1,D_\2 )$
(defined in {\rm(\ref{eq:D0def})--(\ref{eq:D2def})})
i.e., $\chi(\bar x) = (\chi_{D_0}(\bar x),\ldots,\chi_{D_\2}(\bar x))$
where $\chi_{\ldots}$ denotes the characteristic function of the corresponding set.
Then the $2^n \times 6$ value table  $\overline\chi$ of $\chi$ satisfies the equation
\begin{equation}\label{eq:PS}
D\overline\chi = \overline\chi S
\end{equation}
where $D$ is the adjacency $2^n \times 2^n$ matrix of the hypercube and $S$
is the quotient matrix defined in Theorem~\ref{th:main3} (equation (\ref{eq:PS}) is just a matrix treatment of the definition of an equitable family). Equation (\ref{eq:PS}) yields  (see, e.g., \cite{Kro:struct}) that
the value of $\chi$ in a point $\bar x$ uniquely determine the sum of $\chi$ over the sphere
of every radius $r$ centered in $\bar x$. Clearly, the $i$th element of this vector sum
denotes how many elements of $D_i$ are there at distance $r$ from $\bar x$.
To conclude the validity of (a)
for $(n,3)'''\opt$ % $(2^m-4,3)\opt$  % $(n=2^m-4, 2^{n-m}, 3)$
codes, it remains to note that the value $\chi(\bar x)$
is uniquely determined by the distances $d(\bar x,D_0)$ and $d(\bar x+\overline 1,D_0)$.
(b) is an obvious corollary of (a).

If $C_0$ be an
$(n,4)'''\opt$ % $(2^m-3,4)\opt$ % $(n=2^m-3, 2^{n-m-1}, 4)$
code,
then, as follows from the definition {\rm(\ref{eq:C0def})--(\ref{eq:C2def})} of the partition
$(C_0,C_1,C_2,C_\0,C_\1,C_\2 )$, the distances between
$\bar x$ and $C_0$ and between $\bar x+\overline 1$ and $C_0$
determine the cell $C_i$ containing $\bar x$. By the arguments
similar to the previous case, the weight distribution is also uniquely
determined.

The proofs of (a) for
$(n,3)''\opt$ % $(2^m-3,3)\opt$ % $(n=2^m-3, 2^{n-m}, 3)$
$(n,4)''\opt$ % and $(2^m-2,4)\opt$ % $(n=2^m-2, 2^{n-m-1}, 4)$
codes are similar,
based on the generated equitable partition \cite{Kro:2m-3}.
\end{proof}

Explicit formulas for weight distributions
and weight enumerators of equitable families
(or their real-valued generalizations)
can be found in \cite{Kro:struct}.

%===========================================================================%=======================================================
\section{More properties} \label{s:pro}

A real-valued function on $V(H^n)$ is called {\em $1$-centered} if its sum over every radius-$1$ ball
equals $1$. For example, the characteristic functions of $1$-perfect codes are $\{0,1\}$-valued
$1$-centered functions.
Although there are $(n,3)'''\opt$ codes that cannot be lengthened to $1$-perfect codes length $n+3$,
the characteristic function of every such code occurs as a subfunction of $\{0,\frac 13,1\}$-valued
$1$-centered function on $V(H^{n+3})$:
\begin{corollary}
  For every $(n,3)'''\opt$ code $C_0$, the function $f:V(H^{n+3})\to\{0,\frac 13,1\}$ defined as follows is $1$-centered:
  $$
  \begin{array}{lllllll}
  f(\bar x 000) & = & \chi_{C_0}(\bar x), & &
  f(\bar x 111) & = & \chi_{C_\0}(\bar x), \\
  f(\bar x 001) & = & f(\bar x 010) & = & f(\bar x 100) & = & \chi_{C_2}(\bar x)/3, \\
  f(\bar x 110) & = & f(\bar x 101) & = & f(\bar x 011) & = & \chi_{C_\2}(\bar x)/3,
  \end{array}
  $$
  where $C_\0$, $C_2$, $C_\2$ are defined in {\rm (\ref{eq:C0def})--(\ref{eq:C2def})} and $\chi_S$ denotes the characteristic function of a set $S$.
\end{corollary}
The proof consists of straightforward checking the definition by utilizing the array (\ref{eq:tab}).
This embedding result makes some facts known for centered functions (see, e.g., \cite{AvgVas:2006:testing}) applicable for studying $(n,3)'''\opt$ codes
(for similar embedding result for $(n,3)''\opt$ codes, see \cite[Section~4]{Kro:2m-3}).

It is worth to mention here another important common property of the considered classes of codes,
which also can be derived from the results above, but actually has a more direct prove, found in \cite{KOP:2011}.
\begin{theorem}[\cite{KOP:2011}]\label{th:OA}
  Every $(n,3)''\opt$, $(n,4)''\opt$, $(n,3)'''\opt$, or $(n,4)'''\opt$ code $C$ forms
  an {\em orthogonal array} of strength $t=\frac{n-3}2$, $t=\frac{n-4}2$, $t=\frac{n-4}2$, $t=\frac{n-5}2$ respectively;
  that is, for every $t$ coordinates and every values of these coordinates, there are exactly
  $|C|/2^t$ codewords that contain the given values in the given coordinates.
  In an equivalent terminology, the characteristic function of $C$ is {\em correlation immune} of degree $t$.
\end{theorem}
Note that the similar property of $(n,3)\opt$ and $(n,4)\opt$ codes is well known ($t=\frac{n-1}2$, $t=\frac{n-2}2$);
for $(n,3)'\opt$ and $(n,4)'\opt$ codes it also trivially holds ($t=\frac{n-2}2$, $t=\frac{n-3}2$)
 because they can be lengthened to $(n+1,3)\opt$ and $(n+1,4)\opt$, respectively.

%===================================================
% \bibliographystyle{plain}%{spbasic}
% \bibliography{../k}

\providecommand\href[2]{#2} \providecommand\url[1]{\href{#1}{#1}}
  \providecommand\bblmay{May} \providecommand\bbloct{October}
  \providecommand\bblsep{September} \def\DOI#1{{\small {DOI}:
  \href{http://dx.doi.org/#1}{#1}}}\def\DOIURL#1#2{{\small{DOI}:
  \href{http://dx.doi.org/#2}{#1}}}\providecommand\bbljun{June}

\end{document}